\documentclass[sigconf]{acmart}
\pdfoutput=1            
\usepackage[T1]{fontenc}
\usepackage[utf8]{inputenc} 

\usepackage{shortcutsg}  
\usepackage{fancybox}
\usepackage{url}
\usepackage{amsthm} 
\usepackage{balance}

\AtBeginDocument{%
  }

\copyrightyear{2025}
\acmYear{2025}
\setcopyright{cc}
\setcctype{by}
\acmConference[KDD '25]{Proceedings of the 31st ACM SIGKDD Conference on Knowledge Discovery and Data Mining V.2}{August 3--7, 2025}{Toronto, ON, Canada}
\acmBooktitle{Proceedings of the 31st ACM SIGKDD Conference on Knowledge Discovery and Data Mining V.2 (KDD '25), August 3--7, 2025, Toronto, ON, Canada}
\acmDOI{10.1145/3711896.3737217}
\acmISBN{979-8-4007-1454-2/2025/08}

\makeatletter
\gdef\@copyrightpermission{
  \begin{minipage}{0.3\columnwidth}
   \href{https://creativecommons.org/licenses/by/4.0/}{\includegraphics[width=0.90\textwidth]{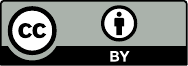}}
  \end{minipage}\hfill
  \begin{minipage}{0.7\columnwidth}
   \href{https://creativecommons.org/licenses/by/4.0/}{This work is licensed under a Creative Commons Attribution International 4.0 License.}
  \end{minipage}
  \vspace{5pt}
}
\makeatother

\title{Evaluating Decision Rules Across Many Weak Experiments}

\author{Winston Chou}
\authornote{Corresponding author.}
\affiliation{
  \institution{Netflix}
  \city{Los Gatos}
  \state{CA}
  \country{USA}
}
\email{wchou@netflix.com}

\author{Colin Gray}
\affiliation{
  \institution{Netflix}
  \city{Los Gatos}
  \state{CA}
  \country{USA}
}
\email{coling@netflix.com}

\author{Nathan Kallus}
\affiliation{
    \institution{Netflix}
    \city{Los Gatos}
    \state{CA}
    \country{United States}
}
\affiliation{
    \institution{Cornell University}
    \city{New York}
    \state{NY}
    \country{United States}
}
\email{nkallus@netflix.com}

\author{Aur\'elien Bibaut}
\affiliation{
  \institution{Netflix}
  \city{Los Gatos}
  \state{CA}
  \country{USA}
}
\email{abibaut@netflix.com}

\author{Simon Ejdemyr}
\affiliation{
  \institution{Netflix}
  \city{Los Gatos}
  \state{CA}
  \country{USA}
}
\email{sejdemyr@netflix.com}

\date{February 2025}


\usepackage{natbib}
\newtheorem{theorem}{Theorem}[section]

\begin{document}

\begin{abstract}
Technology firms conduct randomized controlled experiments (``A/B tests'') to learn which actions to take to improve business outcomes.  In firms with mature experimentation platforms, experimentation programs can consist of many thousands of tests.  To effectively scale experimentation, firms rely on \emph{decision rules}: standard operating procedures for mapping the results of an experiment to a choice of treatment arm to launch to the general user population.  Despite the critical role of decision rules in translating experimentation into business decisions, rigorous guidance on how to evaluate and choose decision rules is scarce.  This paper proposes to evaluate decision rules based on their cumulative returns to business north star metrics.  Although intuitive and easy to explain to decision-makers, this quantity can be difficult to estimate, especially when experiments have weak signal-to-noise ratios.  We develop a cross-validation estimator that is much less biased than the naive plug-in estimator under conditions realistic to digital experimentation.  We demonstrate the efficacy of our approach via a case study of 123 historical A/B tests at Netflix, where we used it to show that a new decision rule would have increased cumulative returns to the north star metric by an estimated $33\%$, directly leading to the adoption of the new rule.
\end{abstract}

\maketitle

\section{Introduction}

Online randomized controlled experimentation (aka A/B testing) is the workhorse of product innovation in many technology companies \citep{deng2017b, thomke2020building}.  Netflix is no exception to this rule, running thousands of A/B tests per year on many diverse parts of the business \citep{ejdemyr2024estimating}.  This scale necessitates the adoption of norms and practices that are easily communicated between scientists, engineers, and product managers alike \citep{hbr2025experimentation}.

In particular, companies often rely on \emph{decision rules}: standard operating procedures that map the outcomes of an experiment to a ``winning" treatment arm to launch to the general user population.  Letting $Y$ denote a \emph{north star} business reward, such as long-term user retention or revenue, a simple example of a decision rule is:
\begin{equation}
    \texttt{\small Launch the treatment arm with the largest effect on $Y$.}
    \label{eqn:rule-simple}
\end{equation}

In practice, decision rules are often substantially more complicated than (\ref{eqn:rule-simple}).  For example, one may have:
\begin{itemize}
    \item \textbf{Statistical significance thresholds}, which dictate that treatments are only eligible for launch if their effect on $Y$ is statistically distinguishable from zero.
    \item \textbf{Proxy metrics}, which dictate that launch decisions instead rely on impacts to a blessed metric $S$ that is positively correlated with $Y$, but has a higher signal-to-noise ratio and/or is easier to measure in short, sample-constrained experiments \citep{bibaut2024learning,richardson2023pareto,tripuraneni2024choosing}. This also covers the use of \textit{surrogate indices}, which combine multiple proxy metrics \textbf{S} into a scalar-valued function $f(\text{\textbf{S}})$ \citep{athey2019surrogate}.
    \item \textbf{Guardrail metrics}, which only permit launches if they do not show a decline in a metric that is negatively correlated with $Y$, or which is important to minimize for other reasons (e.g., app crashes, customer service contacts) \citep{deng2016data}.
\end{itemize}

How should organizations evaluate and choose between these (and many other possible) decision rules? We propose evaluating decision rules based on their \emph{cumulative returns}, which are not measurable in any one experiment, but can be estimated across many historical experiments \citep{ejdemyr2024estimating}.  Intuitively, our method evaluates a decision rule by asking, ``What would the cumulative impact on $Y$ have been if all experiments were decided using this rule?'' 

Although this quantity is intuitive and easy to explain to decision-makers, simple backward-looking estimates of cumulative returns can be severely biased due to a \textit{winner's curse}: since winning arms are sometimes chosen due to noise, the underlying treatment effects of winning arms are likely smaller than the raw historical lifts suggest \citep{lee2018winner, ejdemyr2024estimating}.  A consequence of the winner's curse is that rules that exploit spurious correlations can seemingly perform well in past experiments, but generalize poorly to future experiments.  This bias persists with even an infinite number of experiments so long as each experiment has a finite sample size, and is exacerbated by the low signal-to-noise ratios endemic to digital experimentation \citep{larsen2024statistical, kohavi2014seven}.

To overcome the winner's curse, we demonstrate theoretically and empirically that a cross-validation estimator using experiment splitting \citep[][]{coey2019improving} is many times less biased than the plug-in estimator under conditions realistic to digital experimentation. The cross-validation estimator separates the data used to choose the winning arm from the data used to evaluate its returns, eliminating the winner's curse. The tradeoff is a small negative bias that is easily mitigated in practice.

Our work is closely related to research on optimizing the returns to industrial-scale experimentation \citep{azevedo2018b, azevedo2023b, sudijono2024optimizing}.  This literature highlights how conventions borrowed from scientific practice (e.g., launching if and only if $p < 0.05$) generally fail to maximize north star business metrics. For example, if the potential impact of innovations is heavy-tailed, firms should run many \emph{ex ante} underpowered experiments \citep{azevedo2018b} and launch based on much more relaxed p-value thresholds \citep{sudijono2024optimizing}.  On the other hand, a stricter $p$-value threshold may be preferable for firms that associate a high cost with false discoveries \citep{kohavi2024false}. Our work calls for addressing these questions empirically via the evaluation of decision rule candidates across a fixed set of past experiments.  In our experience, evidence of this kind is essential to convincing decision-makers to consider different rules, as opposed to more abstract arguments based on a superpopulation of hypothetical experiments.

Our work also relates to the literature on proxy metrics in experimentation.  Previous work has taken two distinct approaches to constructing proxy metrics. The first approach assesses the degree to which the proxy and north star metrics agree or disagree in sign, possibly accounting for statistical significance \citep{richardson2023pareto, tripuraneni2024choosing, jeunen2024learning}. This reflects two desirable properties of proxy metrics, namely directional alignment and statistical sensitivity, with less emphasis on an underlying structural model. The second approach targets specific structural estimands based on the north star metric: \cite{bibaut2024learning} develops estimators for the covariance in true treatment effects between the proxy and north star (as opposed to the covariance of estimated treatment effects, which is subject to spurious finite-sample correlations), while \cite{athey2019surrogate} proposes to construct an index of intermediate outcomes to target the unobserved lift in the north star metric. Taking a more general perspective, we treat proxy metrics as a special case of decision rule and propose to estimate their performance empirically, although our methodology also ensures robustness to spuriously correlated noise.

The remainder of the paper is structured as follows:
\begin{itemize}
    \item Section 2 outlines a framework for evaluating decision rules based on their cumulative returns.  Our framework is suited to many practical applications, including the choice of $p$-value thresholds \citep{kohavi2024false, sudijono2024optimizing} and proxy metrics or surrogate models \citep{athey2019surrogate, bibaut2024learning, richardson2023pareto, tripuraneni2024choosing, jeunen2024learning, kharitonov2017learning}.
    \item Section 3 provides an estimator for cumulative returns using past A/B tests and specifies the conditions under which our estimator consistently selects the best rule from a finite set of candidates as the number of experiments (not the number of units in each experiment) increases.
    \item Section 4 illustrates the benefits of our framework in a stylized example. We use our method to evaluate and select between candidate proxy metrics in simulated data and demonstrate the favorable properties of our approach.
    \item Lastly, Section 5 describes the real-world application of our framework to select decision rules at Netflix. Using our method, we demonstrated that a new decision rule would increase the cumulative returns on a sample of past experiments by approximately 33\%.  This evidence led directly to the adoption of our rule to decide new experiments at Netflix.
\end{itemize}

\section{Framework and Notation}
\label{sec:setup}
We study a setting in which a firm has conducted $N$ past experiments indexed by $i = 1, \ldots, N$.  Experiment $i$ has $K_i$ arms, indexed by $k = 1, \ldots, K_i$, with each arm consisting of $M_i$ units, indexed by $m = 1, \ldots, M$.  We typically think of $M_i$ as large in online experiments, numbering in the thousands or even millions.  However, $M_i$ is also bounded for even the largest firms, which have a finite user base to allocate across many parallel experiments \citep{tang2010overlapping}.  Therefore, as in \cite{bibaut2024learning}, our theoretical results focus on the asymptotic regime in which the number of \emph{experiments}, not the number of \emph{users}, goes to infinity.

Our observations consist of a vector of $J$ outcomes $O^i_{km}\in\mathbb R^{J}$ for each unit in each arm of each experiment.  We assume $O^i_{km}$ are independent across all $i,k,m$, while $O^i_{k1},\dots,O^i_{kM_i}$ are identically distributed for each $i,k$. For notational ease, we append vectors $O^i_{km}$ into a single experiment-level vector $\mathbf O^i\in\mathbb R^{K_i\times M_i\times J}$.

A decision rule $D(\mathbf O^i)\in\{1,\dots,K\}$ maps observations from an experiment to a choice of arm. Note that $D$ is implicitly a function of the number of observations $M_i$, encoded in the shape of $\mathbf O^i$. Thus, for example, the decision rule can compute standard errors and account for statistical significance.

We assume there is a known reward function $\psi:\mathbb R^J\to\mathbb R$, and let: $$R_{ik}=\mathbb E[\psi(O^i_{k1})]$$ be the mean reward from arm $k$ in experiment $i$.  For example, if of the $J$ metrics we treat the first one as the reward, we set $\psi(o)=o_1$. The other $J-1$ metrics may still be used (or not) for making a decision.\footnote{For example, the second metric may be a proxy with lower variance whose mean might have some bias to the true reward. Or, the other metrics may be pretreatment variables used for variance reduction via post-stratification or regression adjustment.}

We are interested in the reward generated by a decision rule, averaged over the $N$ past experiments:\footnote{We use weighted averages in practice, but omit the weights here for exposition.}
$$
\Gamma(D)= \frac1N\sum_{i=1}^NR_{iD(\mathbf O^i)},
$$
We propose elevating decision rules that do well in terms of this metric. Specifically, among a finite set of decision rules ${\mathcal D}$, we want to identify $D^* = \text{argmax}_{D \in \mathcal D} R(D)$.

To build intuition, we illustrate some common decision rules. A simple decision rule selects one outcome $\tilde{O}$ and launches the arm with the largest sample mean:
$$D_1(\mathbf O^i)=\argmax_{k\in[K_i]} \frac{1}{M_{ik}} \sum_{m=1}^{M_{ik}} \tilde{O}^i_{km}$$

A slightly more complex rule defines a multivariate function $\phi:\mathbb R^J\to\mathbb R$ that blends multiple outcome metrics and maximizes the mean of that value. This covers cases such as proxy metrics or shrinkage estimators:
$$D_2(\mathbf O^i)=\argmax_{k\in[K_i]} \frac{1}{M_{ik}} \sum_{m=1}^{M_{ik}} \phi(O^i_{km})$$

Lastly, it is also common to consider only launching arms for which $\phi$ is statistically significantly greater than a reference arm. Denoting the reference arm as $k=1$ and the set of such arms as $\mathcal S(\mathbf O^i)$, this rule is written as:

\[
D_3(\mathbf O^i)=
\begin{cases} 
 \argmax_{k\in[\mathcal S(\mathbf O^i)]} \frac{1}{M_{ik}} \sum_{m=1}^{M_{ik}} \phi(O^i_{km}) & \text{ if } |\mathcal S(\mathbf O^i)| > 0 \\ 
1 & \text{ if } |\mathcal S(\mathbf O^i)| = 0.
\end{cases}
\]

We can envision many other variants and extensions, including correcting for pretreatment variables or using lower confidence bounds.  In any event, given the many reasonable-sounding alternatives, we want to avoid choosing rules arbitrarily, and instead choose between them in a data-driven manner. 

\section{Proposed Methodology}
\label{sec:method}

As a baseline, we first consider the following ``naive'' plug-in estimator for $R_{iD(\mathbf{O}^i)}$:
\begin{equation}
    \hat{R}_i^{naive}(D) = \frac{1}{M_i} \sum_{m=1}^{M_i} \sum_{k=1}^{K_i} \mathbf{1}(D(\mathbf{O}^i) = k) \psi(O^i_{km}). 
\end{equation}

$\hat{R}_i^{naive}$ suffers from the well-known winner's curse bias:
\begin{eqnarray}
    && \mathbb{E}[\hat{R}_i^{naive}] - R_{iD(\mathbf{O}^i)} \\\nonumber && \qquad = \sum_{k=1}^{K_i} \Pr(D(\mathbf{O}^i) = k) (\mathbb{E}[\psi(O^i_{k1}) | D(\mathbf{O}^i) = k] - R_{ik}),
\end{eqnarray}
which arises because the conditional expectation of an arm's reward given that is selected is not the same as its unconditional expectation \citep{lee2018winner}.  The winner's curse is vanishing in $M$ (the number of units per experiment) rather than $N$ (the number of experiments), so accumulating more experiments will only yield more precise estimates of the wrong estimand rather than eliminate the bias.

Alternatively, our proposed cross-validation estimator uses experiment splitting to eliminate the winner's curse:

\begin{center}
\shadowbox{%
  \begin{minipage}{0.95\linewidth}
    \textsc{Cross-Validation Algorithm}
    \begin{itemize}
        \item Randomly split the units in experiment $i$ into $P_i$ folds.
        \item For $i = 1, \ldots, N$ and $p = 1, \ldots, P_i$:
        \begin{itemize}
            \item Apply the decision rule $D(\mathbf{O}_{-p}^i)$ to select a winning arm.
            \item Estimate the reward of $D(\mathbf{O}_{-p}^i)$ on the $p$th fold as:
            \begin{equation*}
                \hat{R}^{CV}_{ip} = \frac{1}{M_{ip}} \sum_{m = 1}^{M_{ip}} \sum_{k=1}^{K_i} \mathbf{1}(D(\mathbf{O}_{-p}^i) = k) \psi(O^i_{km})
            \end{equation*}
            where $M_{ip}$ is the number of units in fold $p$ for experiment $i$.
        \end{itemize}
        \item Lastly, estimate the mean reward as the (weighted) mean of $\hat{R}^{CV}_{ip}$ over all folds and experiments:
            \begin{eqnarray*}
                \hat{\Gamma}^{CV} &=& \frac{1}{\sum_{i=1}^NP_i} \sum_{i=1}^N\sum_{p=1}^{P_i}  \hat{R}^{CV}_{ip} 
            \end{eqnarray*}
    \end{itemize}
  \end{minipage}%
}
\end{center}

\medskip

In a nutshell, our estimator separates the data used to choose the winning arm from the data used to estimate the reward of that arm, eliminating the winner's curse.  Note, however, that the cross-validation estimator is also biased because the decision rule is no longer applied to the full sample, changing the probability with which each arm is chosen:
\begin{eqnarray}
    && \mathbb{E}[\hat{R}^{CV}_{ip}] - R_{iD(\mathbf{O}^i)} = \\\nonumber
    && \qquad \sum_{k=1}^K \left(\frac1P\sum_{p=1}^P\Pr(D(\mathbf{O}_{-p}^i) = k) - \Pr(D(\mathbf{O}^i) = k)\right) R_{ik}.
\end{eqnarray}

There are two ways to deal with the remaining bias in the CV estimator.

First, we can increase the number of folds $P$. While the CV estimator is biased for the expected reward of the rule applied to the full data, when all folds are of the same size $M_{ip}=M_i/P$ then it is unbiased for the expected reward of the rule applied to the first $(P - 1) / P$ fraction of the data. These two estimands will become very close as the number of folds $P$ approaches the sample size $M_i$. However, in the best case we can use leave-one-out cross-validation to test the rule that uses $M_i-1$ observations, which is still technically different from using all $M_i$ observations. Although this difference should be very small in applications, it can be meaningful when experiments are just barely powered to detect effects between arms.  In practice, we recommend assessing the sensitivity of the CV estimator to the choice of $P$.

Second, we can provide theoretical guarantees in the specific setting where the amount of data in each experiment is itself random. We show that, up to a scaling, cross-validation incurs zero bias on average in this realistic setting.  Specifically, suppose that $M_i$ is random and Poisson-distributed. This is consistent, for example, with enrolling a fixed percentage of all incoming units into each experiment over a fixed time period. Consider leave-$\ell$-out cross-validation, where there are $P_i=\ell!^{-1} \prod_{j=0}^{\ell-1}(M_i-j)$ folds consisting of all subsets of $M_{ip}=\ell$ units $\{1,\dots,M_i\}$.

\begin{theorem}\label{thm: unbiased cv}
Suppose $M_i\sim\mathrm{Poisson}(M_0)$.  Then, for the leave-$\ell$-out CV estimator, we have $\mathbb E[M_0^{-\ell}\ell!\sum_{p=1}^{P_i} \hat{R}^{CV}_{ip}]=\mathbb E[R_{iD(\mathbf O^i)}]$.
\end{theorem}

\begin{proof}
Define $R_{i}(t)=\mathbb E[R_{iD(\mathbf O^i)}\mid M_i=t]$. Then $\mathbb E[\hat{R}^{CV}_{ip}\mid M_i]=R_{i}(M_i-\ell)$, so $\mathbb E[\sum_{p=1}^{P_i} \hat{R}^{CV}_{ip}\mid M_i]=P_iR_{i}(M_i-\ell)$. On the other hand, $\mathbb E[R_{iD(\mathbf O^i)}]=\mathbb E[R_{i}(M_i)]$ by total expectation. The Stein-Chen identity \citep{chen} states that, for $X\sim\mathrm{Poisson}(\lambda)$, we have $\lambda\mathbb E[f(X)]=\mathbb E[X f(X-1)]$ for any $f:\mathbb Z\to\mathbb R$. Applying this identity $\ell$ times we obtain $\lambda^\ell\mathbb E[f(X)]=\mathbb E[X(X-1)\cdots(X-\ell+1)f(X-\ell)]$. Since $P_i=\ell!^{-1} \prod_{j=0}^{\ell-1}(M_i-j)$, we obtain the conclusion.
\end{proof}

The practical implication of Theorem~\ref{thm: unbiased cv} is that, if $M_i$ is Poisson, the cross-validation estimator will be able to select the best rule (in terms of $R$) from a finite set of decision rules as the number of experiments $N$ goes to infinity.  Importantly, this does not require the size of each experimental arm $M_i \to \infty$, which is an unrealistic asymptotic regime given that all firms have a finite user pool to allocate to many parallel experiments \citep{tang2010overlapping}.  We formalize and prove the result for leave-one-out cross-validation ($P_i = M_i$) below.

\begin{theorem}
Suppose $M_i\sim\mathrm{Poisson}(M_0)$ and $\|\psi(O_{km}^i)\|_\infty<\infty$. Fix $\abs{\mathcal D}<\infty$ and let: $$\textstyle\hat D\in\argmax_{D\in\mathcal D}\hat \Gamma^{CV}(D)$$ with $\ell = 1$.  Then $\max_{D\in\mathcal D}\Gamma(D)-\Gamma(\hat D)=O_p\left(\sqrt{ \log|\mathcal D|/(NM_0)}\right)$.
\end{theorem}
\begin{proof}
Let $X_i(D)=M_0^{-1}\sum_{m=1}^{M_i}\hat R_{ip}^{CV}(D)$ and observe that $\hat D\in\argmax_{D\in\mathcal D}\frac1N\sum_{i=1}^N X_i(D)$. Therefore, $\max_{D\in\mathcal D}\Gamma(D)-\Gamma(\hat D)\leq 2\sup_{D\in\mathcal D}\left|\frac1N\sum_{i=1}^N \bar X_i(D)\right|$ where $\bar X_i(D)= X_i(D) - R_{iD(\mathbf O^i)}$.  Theorem \ref{thm: unbiased cv} established that $\mathbb E \bar X_i(D)=0$. Let $C=\|\psi(O_{km}^i)\|_\infty<\infty$.  Then:
\begin{align*}
    & \mathbb{E}\left[\exp\left(\frac{t}{N}\sum_{i=1}^N \bar{X}_i(D)\right)\right] \\
    &\quad = \left(\mathbb{E}\left[\exp\left(\frac{t}{N}\bar{X}_i(D)\right)\right]\right)^N \\
    &\quad = \left(e^{-s\Gamma(D)} \mathbb{E}\left[\exp\left(s X_i(D)\right)\right]\right)^N \bigg|_{s=t/N} \\
    &\quad = \left(e^{-s\Gamma(D)} \exp\left(M_0\left( \mathbb{E}\left[e^{(s/M_0) \hat{R}_{ip}^{CV}}\right] - 1\right)\right)\right)^N \bigg|_{s=t/N} \\
    &\quad \le \left(\exp\left(\frac{s^2\mathbb{E}[(\hat{R}_{ip}^{CV})^2]}{M_0}\right)\right)^N \bigg|_{s=t/N} \\
    &\quad \le \exp\left(\frac{t^2 C^2}{NM_0}\right)
\end{align*}
for $|s|C \leq M_0$.  The rest follows by a Chernoff and union bound.
\end{proof}

\section{Example: Choosing a Proxy Metric from Past Experiments}
\label{sec:proxy}

Above, we developed notation to formalize a target estimand (expected reward), demonstrated that an estimator based on cross-validation does not suffer from the winner's curse, and discussed ways to counteract the remaining bias in the CV estimator. This section further builds intuition by focusing on a specific application of our method: choosing a good proxy metric for deciding experiments \citep{bibaut2024learning,tripuraneni2024choosing,richardson2023pareto}.  

Proxy metrics are commonly used in digital experimentation to accelerate decision-making by replacing statistically insensitive and/or long-term metrics (e.g., user retention) with more sensitive and/or short-term metrics (e.g., user engagement). However, choosing a good proxy metric is non-trivial as interventions that increase the proxy metric may not increase the north star metric, even if the two variables are strongly correlated in either experimental \citep{bibaut2024learning} or observational \citep{athey2019surrogate} data.  To decide between different proxy metrics, we can apply the above estimators for the expected reward associated with each proxy and choose the one(s) which leads to the largest estimated reward.

\subsection{Setup}

To simplify our numerical example, we assume $N$ past experiments with $K = 2$ arms each and $M_i = M$ units per arm.  In each experiment, we estimate treatment effects $\hat{\tau}_Y^i$ and $\hat{\tau}_S^i$ on the north star and proxy metric, respectively.  We assume that $M$ is large and that our treatment effect estimator is the difference in sample means, so that the observed treatment effects are Gaussian conditional on the true treatment effects:
\begin{equation}
    (\hat{\tau}_Y^i, \hat{\tau}_S^i) \sim \mathcal{N}\left((\tau_Y^i, \tau_S^i), \frac{2}{M}\Omega\right),
\end{equation}
where $\Omega$ is the unit-level sampling error, which we assume is homoskedastic across arms.

Importantly, we do not assume that the covariance matrix $\Omega$ is diagonal.  In general, $Y$ and $S$ will be correlated at the unit level outside of any particular experiment \citep{cunningham2020interpreting}.  For example, highly engaged users ($S$) are also likely to be more retentive ($Y$), regardless of whether interventions that increase engagement also increase retention.  Thus, we distinguish between the \emph{true} associations between the underlying treatment effects and the \emph{observed} associations between treatment effects estimated in finite experiments \citep{bibaut2024learning}. This finite-sample noise amplifies the winner's curse for proxy metrics and motivates the use of our CV estimator.

We denote the decision rule that utilizes metric $S$ as $D_S$ and measure the quality of $S$ as a proxy for $Y$ using the cumulative reward across past experiments associated with $D_S$.\footnote{While our theoretical analysis focused on the expected reward, here we focus on the cumulative reward, which we find is more intuitive to non-technical stakeholders.} In our notation, these rewards are:
\begin{equation}
N \Gamma(D_S) = \sum_{i=1}^N R_{iD_S(\mathbf O^i)},
\end{equation}
where: \begin{equation}
    R_{iD_S(\mathbf O^i)} = \begin{cases}
        \tau_Y^i & \text{if $\hat{\tau}_S^i > 0$} \\
        0 & \text{otherwise}.
    \end{cases}
\end{equation}
In words, $N \Gamma(D_S)$ is the expected cumulative impact on $Y$ of using a decision rule $D_S$ that launches the treatment so long as the observed treatment effect on $S$ is positive.\footnote{In practice, it may be more common to require the treatment to be \emph{significantly} better than control.  Incorporating this requirement does not qualitatively affect the following analysis, as it simply multiplies the expectations in Equations \ref{eqn:target}-\ref{eqn:cv-target} by a scalar.}  Intuitively, if $N \Gamma(D_S)$ is large, then $S$ is a good proxy for $Y$.

\subsection{Closed-Form Analysis}

For a simple closed-form analysis of $N \Gamma(D_S)$, we assume the true treatment effects across experiments $i$ are also Gaussian and centered on zero:
\begin{equation}
    (\tau_R^i, \tau_S^i) \sim \mathcal{N}\left(\mathbf{0}, \Lambda\right).
\end{equation}
Then, letting $\Lambda_{ij}$ and $\Omega_{ij}$ denote the $ij$'th entry of $\Lambda$ and $\Omega$, the true expected reward is proportional to:
\begin{equation}
    \mathbb{E}[R_{iD_S}] \propto \frac{\Lambda_{12}}{\sqrt{\Lambda_{22} + 2 \Omega_{22}/M}},
    \label{eqn:target}
\end{equation}
which intuitively is increasing in the covariance in true treatment effects.\footnote{See Appendix~\ref{sec:app-closed} for derivations of Equations \ref{eqn:target}-\ref{eqn:cv-target}.}

In contrast, the expectation of the naive estimator is increasing not just in the covariance of true treatment effects, but also in the covariance of the unit-level sampling error:
\begin{equation}
    \mathbb{E}[\hat{R}^{naive}_{iD_S}] \propto \frac{\Lambda_{12} + 2\Omega_{12} / M}{\sqrt{\Lambda_{22} + 2 \Omega_{22}/M}}.
    \label{eqn:naive-target}
\end{equation}
Therefore, the naive estimator can prefer proxy metrics whose correlation with the north star is driven more by noise than by true treatment effects.  A standard example is clickbait: Although highly engaged users may be more retentive outside of any particular experiment ($\Omega_{12}$ large), interventions that merely increase clickbait will not have the desired effect on retention ($\Lambda_{12}$ small).

Note that the risk of bias is greater when the scale of true treatment effects is very small relative to the unit-level sampling variance, as is often the case in digital experiments \citep{larsen2024statistical}.  Thus, the naive estimator will skew towards sensitive ``upper funnel'' metrics such as page visits, as compared to ``lower funnel'' metrics that are harder to move (e.g., signups or checkouts) but more strongly correlated with north star business metrics like revenue.

Lastly, the expectation of the cross-validation estimator is given by:
\begin{equation}
    \mathbb{E}[\hat{R}^{CV}_{iD_S}] \propto \frac{\Lambda_{12}}{\sqrt{\Lambda_{22} + 2 \Omega_{22}/M_p}},
    \label{eqn:cv-target}
\end{equation}
where $M_p = M(P - 1)/P$.  Like the true reward, this expectation is also insensitive to the noise covariance $\Omega_{12}$, depending only on the covariance between the true treatment effects.  It is also negatively biased due to the division in (\ref{eqn:cv-target}) by $M_p < M$, although this can be mitigated by increasing $P$. 

\begin{figure}
    \centering
    \includegraphics[width=0.94\linewidth]{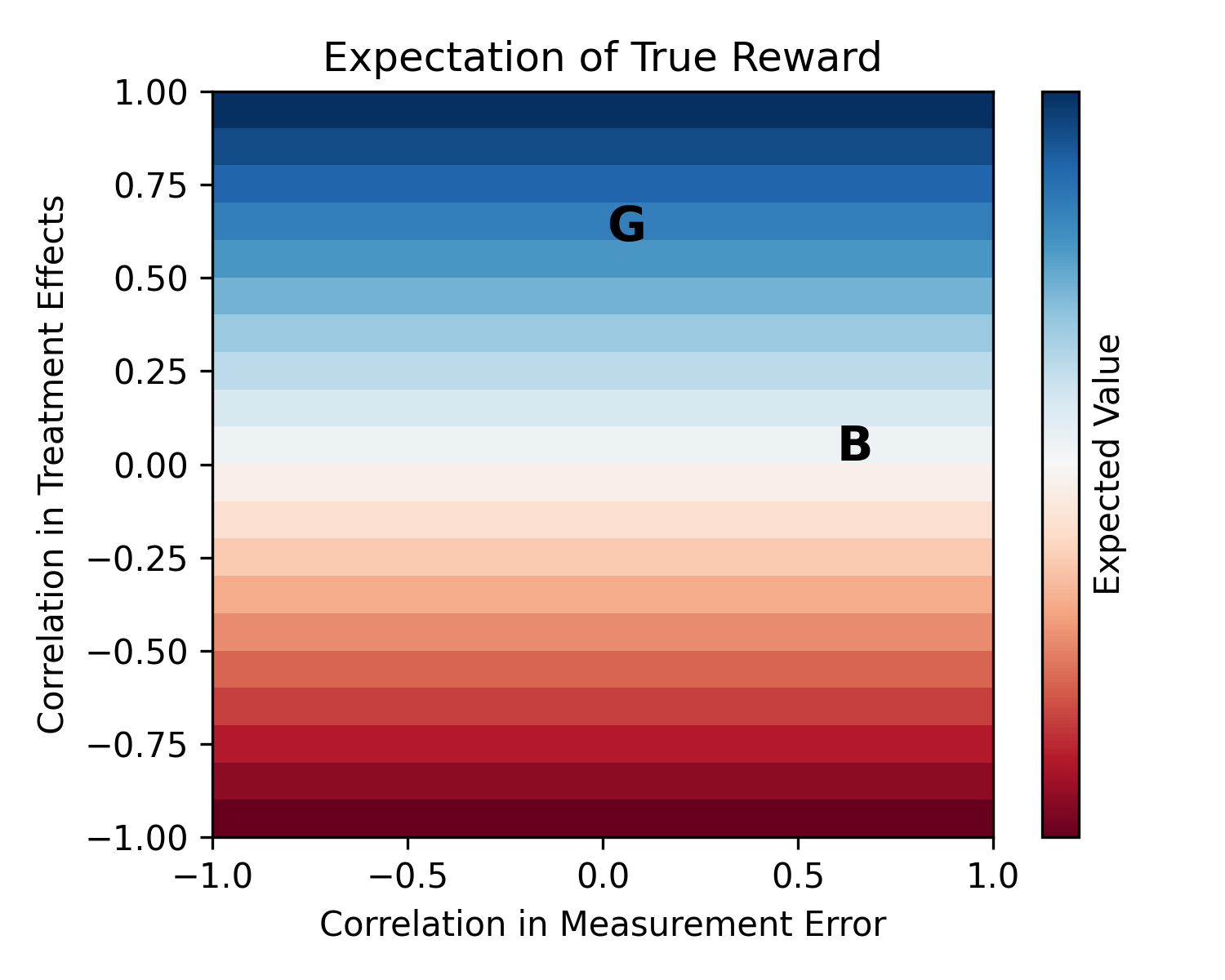}
    \includegraphics[width=0.94\linewidth]{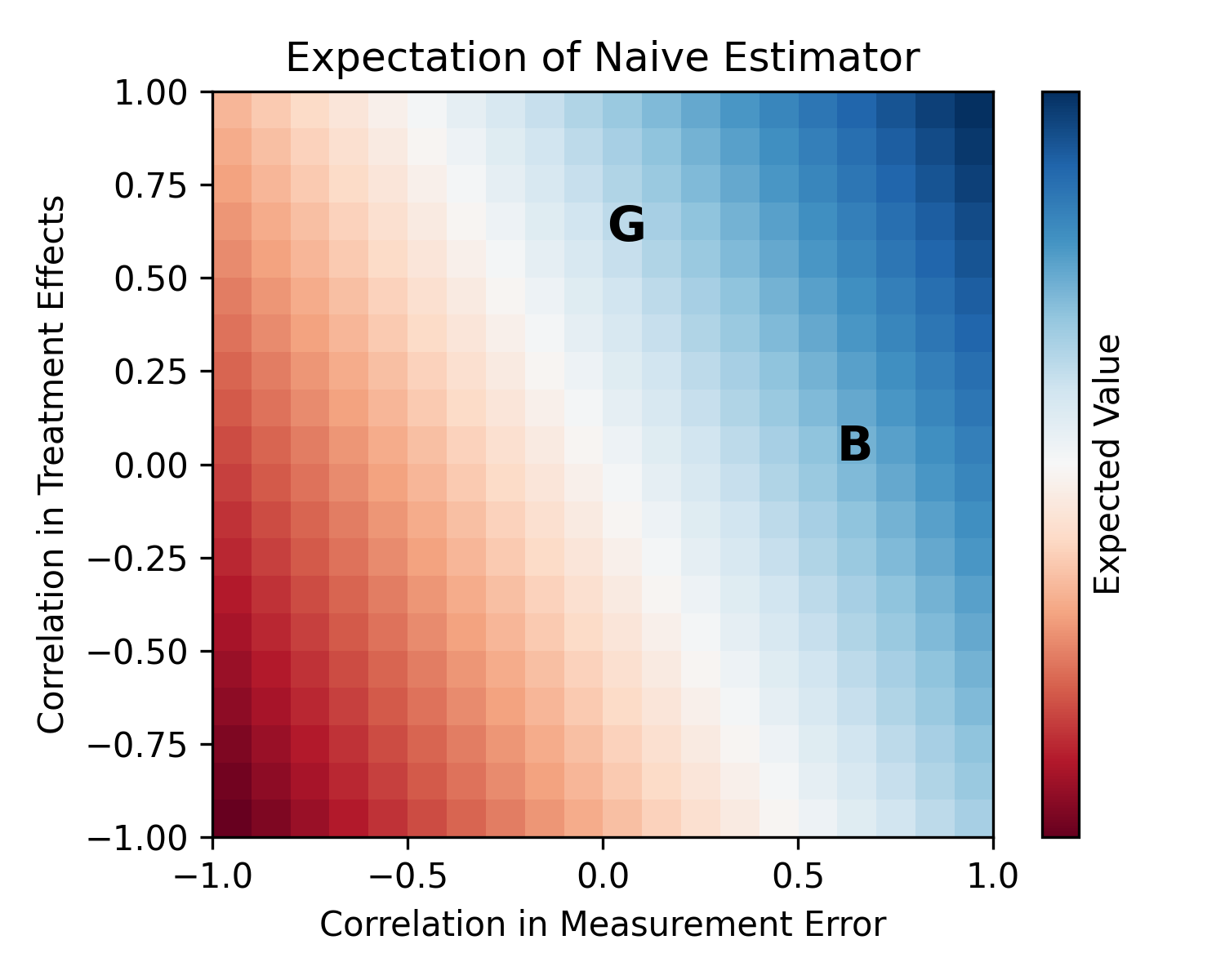}
    \includegraphics[width=0.94\linewidth]{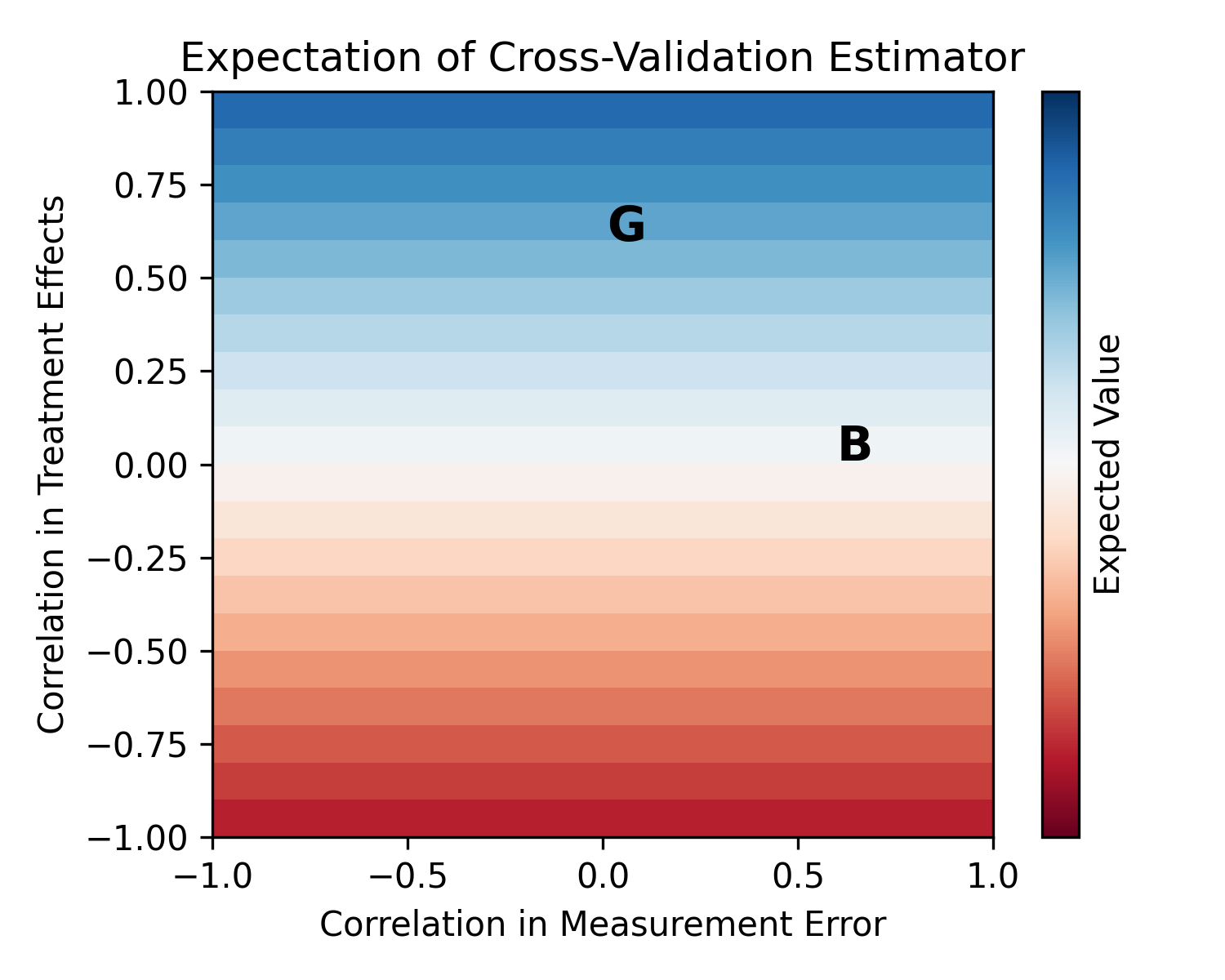}
    \caption{Expected Values of True Reward and Estimators.  While the true expected reward only depends on the covariance of true treatment effects, the naive estimator is sloped and can prefer ``Bad'' proxies to ``Good'' proxies.  The CV estimator is negatively biased (diluted) but still ranks proxies correctly.}
    \label{fig:rewards}
\end{figure}

Figure~\ref{fig:rewards} illustrates these dynamics. In the top panel, we plot the true reward $\mathbb{E}[R_{iD_S}]$ as the correlation between true treatment effects ($y$-axis) and measurement error ($x$-axis) varies from -1 to 1.\footnote{The other parameter values are set to the same values as in our simulation below.}  Regions of the same shade represent level sets (the same value of the reward).  As the panel shows, the true expected reward is only sensitive to the correlation in true treatment effects.

To help frame our subsequent analyses, we also plot the coordinates of a ``\textbf{G}ood'' proxy metric whose treatment effects strongly covary with treatment effects on the north star metric and a ``\textbf{B}ad'' proxy metric whose treatment effects have almost no correlation with treatment effects on the north star.  As expected, the true reward associated with the \textbf{G}ood proxy metric is greater than that of the \textbf{B}ad proxy metric.

In contrast, the middle panel of Figure~\ref{fig:rewards} plots the expected value of the naive estimator as we vary the same correlations.  Due to its sloped level sets, the naive estimator associates a greater reward with the bad proxy metric due to its strongly correlated measurement error.

Lastly, the bottom panel of Figure~\ref{fig:rewards} plots the expected value of our CV estimator with just two folds. The flat level sets demonstrate that the CV estimator will rank correctly based on the true reward, although their shade is diluted due to the secondary source of bias in the CV-estimated reward.  In practice, this bias is easily mitigated by increasing the number of folds. 

\subsection{Simulation Study}

To shed light on additional properties of our estimators, we conduct a simulation study.\footnote{Code to replicate all figures except those in Section~\ref{sec:application} can be found at \url{https://github.com/winston-chou/linear-proxy-metrics}.}  The simulation parameters reflect important qualitative aspects of digital experimentation in settings like ours.\footnote{Their specific numerical values are reported in the Appendix.}  Specifically, we simulate a weak signal-to-noise regime in which the scale of true treatment effects is small relative to the scale of measurement error \citep{larsen2024statistical, kohavi2014seven}.  We also study the effect on our estimators as the number of experiments $N$ and the number of units per experiment $M$ tend to infinity.  The former asymptotic regime is more realistic than the latter: Although companies like Netflix can cheaply run additional experiments, the number of units (e.g., users or subscribers) available for each experiment remains bounded due both to physical constraints and to the need to allocate many experiments in parallel \citep{tang2010overlapping}.

\paragraph{Noise-to-signal ratio.}  As can be seen from Equations~\ref{eqn:target} and~\ref{eqn:naive-target}, the relative bias of the naive estimator is proportional to the ratio of the covariance in measurement error to the covariance in true treatment effects $\Omega_{12} / \Lambda_{12}$.  Letting $\sigma_S$ denote the standard deviation of $S$ and $\sigma_{\tau_S}$ the standard deviation of $\tau_S$, the relative bias is therefore proportional to $\sigma_S / \sigma_{\tau_S}$, the noise-to-signal ratio of $S$.  Therefore, as treatment effects on the proxy metric become infinitely small relative to their measurement error, the relative bias of the naive estimator also tends to infinity.  This is illustrated in Figure~\ref{fig:bias-snr}, where we simulate the bias of the naive estimator as the scale of the noise $\sigma_S$ tends to infinity with $\sigma_{\tau_S}$ fixed.  As the plot shows, the relative bias of the naive estimator is also unbounded.  In contrast, because the CV estimator is biased towards zero, its relative bias is negative for positive rewards and bounded between $-100\%$ and $100\%$.

\begin{figure}[ht]
    \centering
    \includegraphics[width=\linewidth]{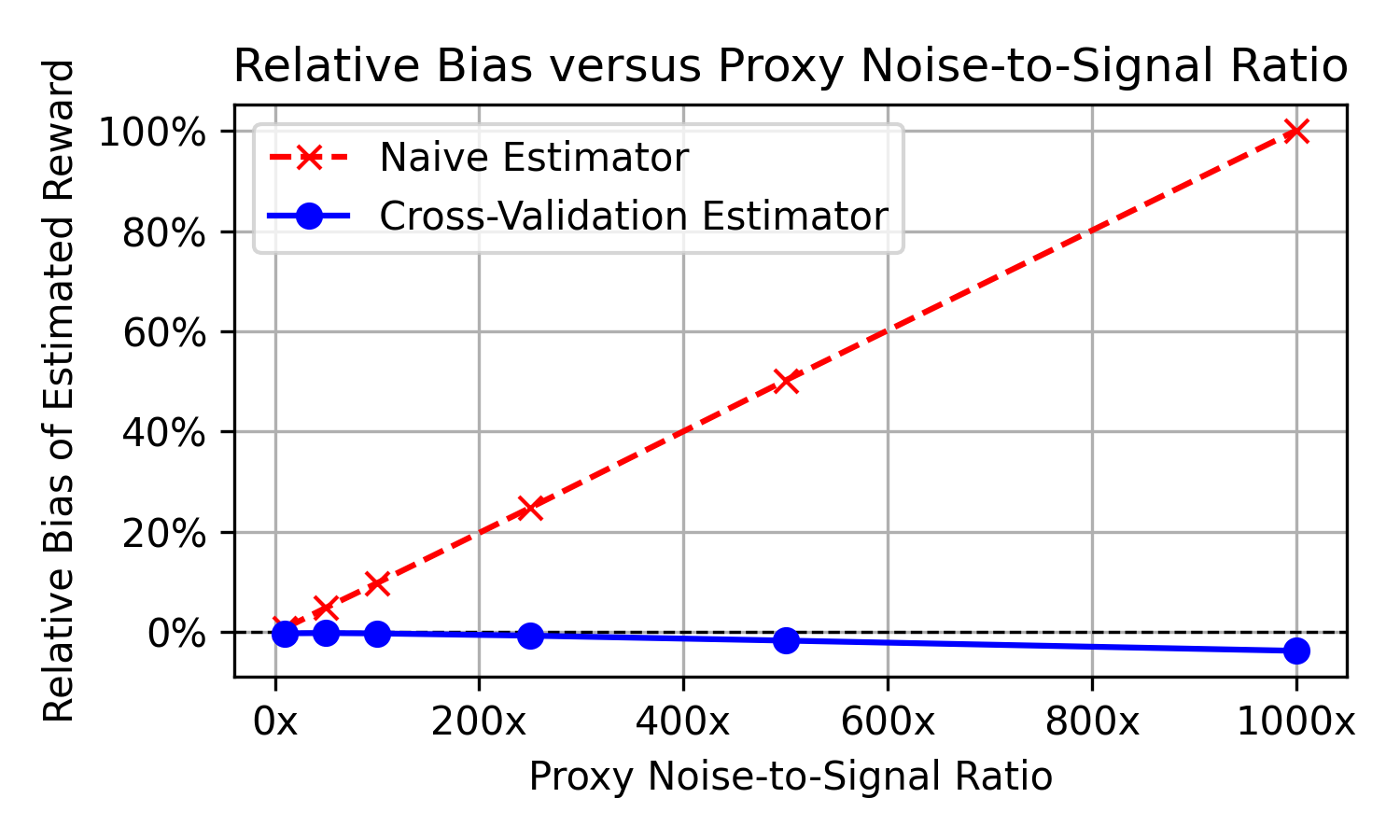}
    \caption{Bias of Naive and CV Estimators vs. Noise-to-Signal Ratio of $S$.  As the scale of the measurement error in $S$ tends to infinity for a fixed scale of \emph{treatment effects on $S$}, the bias of the naive estimator tends to infinity, whereas the relative bias of the CV estimator is bounded.}
    \label{fig:bias-snr}
\end{figure}

\paragraph{Number of units per experiment.}  While the bias of both estimators tends to zero as $M \to \infty$, the CV estimator can be many times less biased at smaller values of $M$.  This is shown in Figure~\ref{fig:bias-m}, where we again contrast the \textbf{G}ood and \textbf{B}ad proxies charted in Figure~\ref{fig:rewards}.  The top panel of Figure~\ref{fig:bias-m} shows the true rewards associated with each proxy metric as we increase $M$.  For reference, we also plot these rewards with light-opacity dashed lines in the other panels.  As the middle panel of Figure~\ref{fig:bias-m} shows, the naive estimator actually favors the \textbf{B}ad proxy at smaller values of $M$ due to its strongly correlated measurement error.  In contrast, the CV estimator correctly chooses the \textbf{G}ood proxy at all values of $M$, although it is very slightly biased towards zero.

\begin{figure}[ht]
    \centering
    \includegraphics[width=\linewidth]{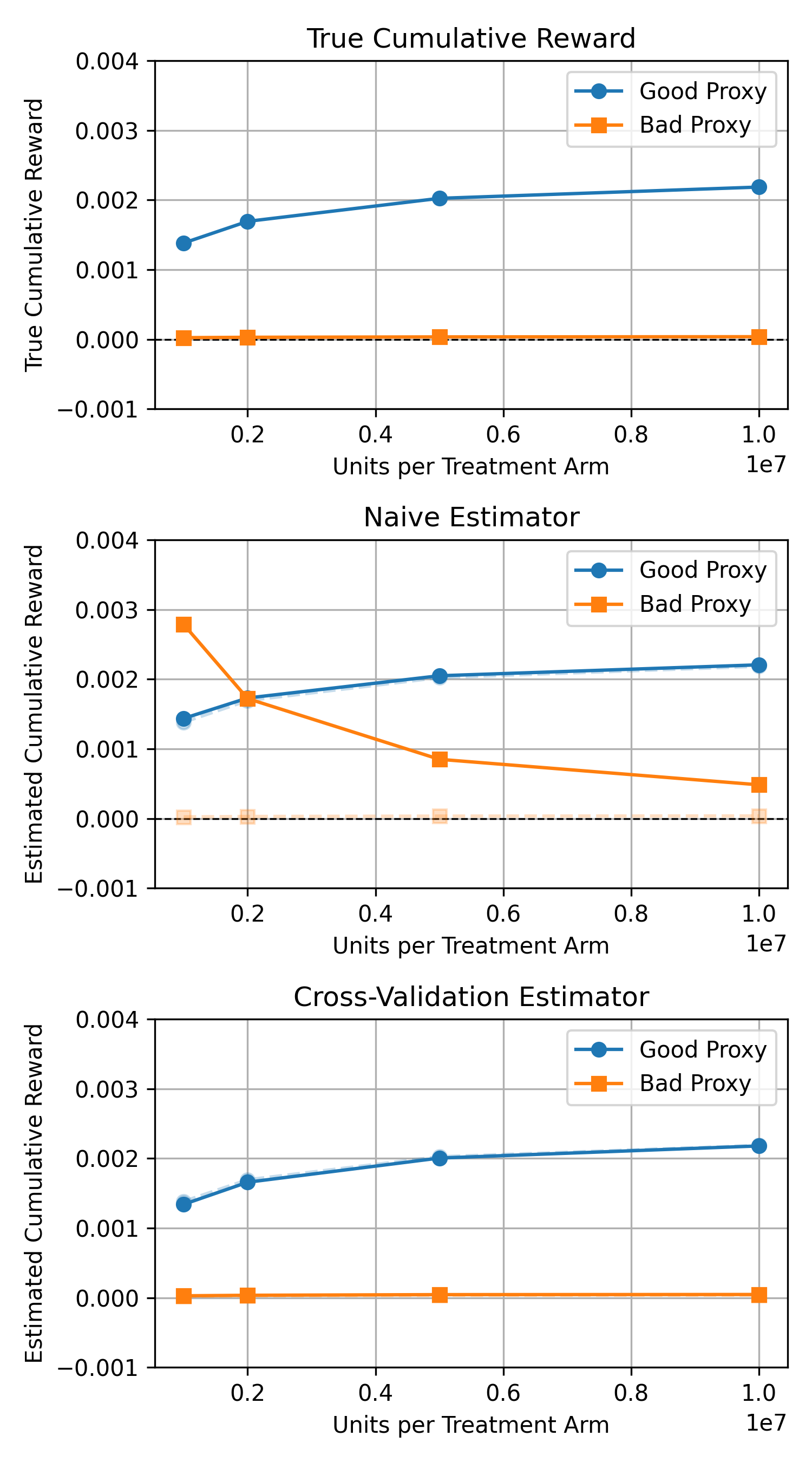}
    \caption{Asymptotics as $M \to \infty$.  The light-opacity dashed lines denote the truth, while the solid lines denote estimators.  While the bias of both estimators goes to zero as the number of units per treatment arm $M \to \infty$, the naive estimator is much more biased at small values of $M$, which can lead to the wrong choice of proxy.}
    \label{fig:bias-m}
\end{figure}

\paragraph{Number of experiments.}  Lastly, Figure~\ref{fig:bias-n} shows how our estimators behave with $M$ fixed but the number of experiments $N \to \infty$, which we consider the more realistic asymptotic regime for our setting.  As before, we plot the true cumulative returns in the topmost panel and then with light-opacity dashed lines in the remaining panels as a reference point.  Note that, with $M$ fixed, increasing $N$ only exaggerates the bias of the naive estimator.  This problem is more severe for the bad proxy due to its greater noise-to-signal ratio.  In contrast, the CV estimator closely tracks the true cumulative reward, although with a slightly pessimistic bias.

\begin{figure}[!ht]
    \centering
    \includegraphics[width=\linewidth]{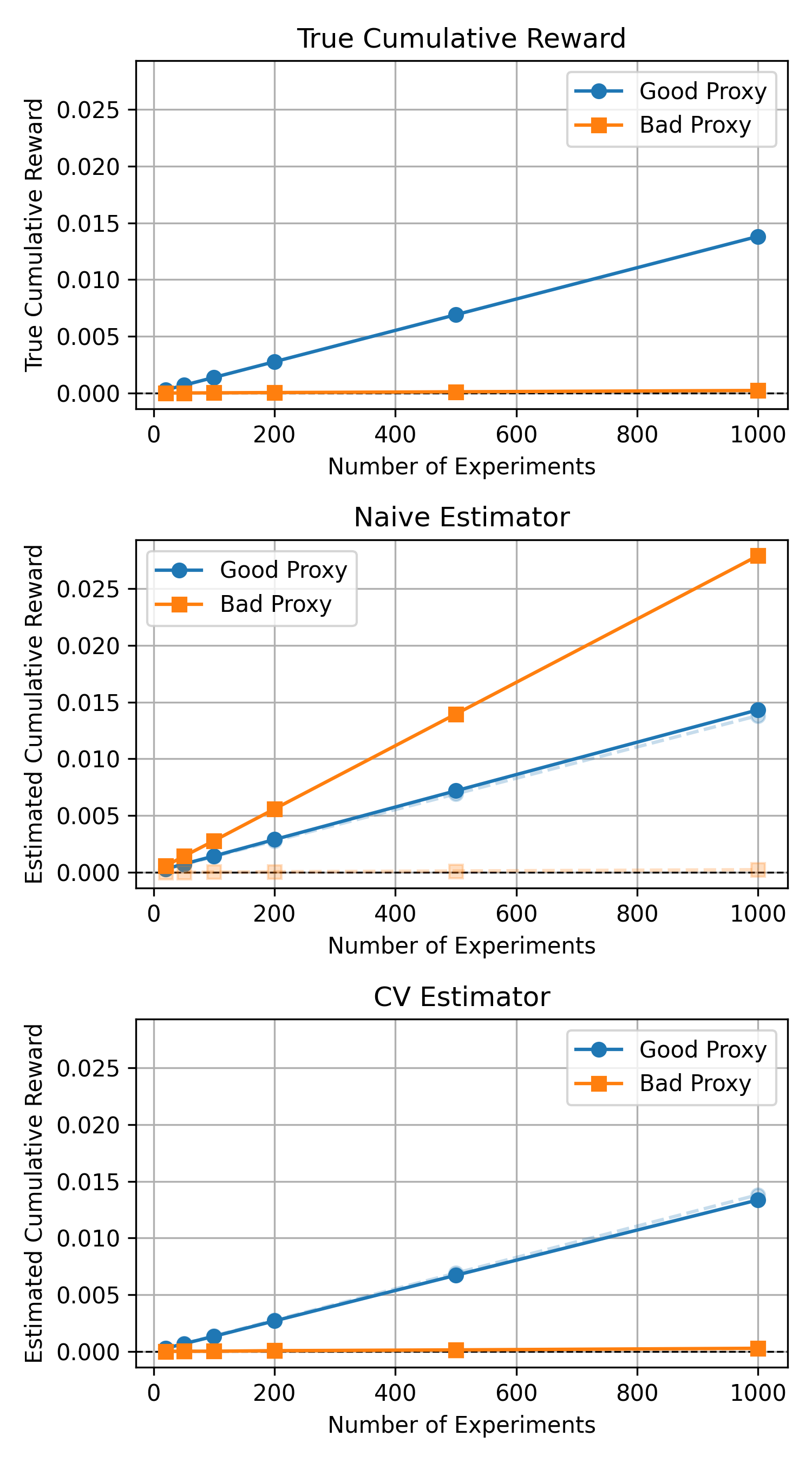}
    \caption{Asymptotics as $N \to \infty$.  The light-opacity dashed lines denote the truth, while the solid lines denote estimators.  Increasing the number of experiments $N$ without increasing the size of each experiment only increases the bias of the naive estimator.}
    \label{fig:bias-n}
\end{figure}

\section{Real-World Application}
\label{sec:application}

In this section, we describe the application of our method to successfully advocate for a new experimentation decision rule at Netflix.  This work is part of a larger project to develop proxy metrics using meta-analysis, initially described in \cite{bibaut2024learning}.  In that paper, we developed methods for fitting linear structural models of treatment effects using data from past experiments.  We deployed these methods to construct an improved proxy metric as a linear combination of component metrics, which include the previous status quo metric. We then used techniques in this paper to demonstrate the favorability of the resulting metric relative to the previous decision metric.  Our method was essential to convincing our stakeholders to adopt the new decision rule.  This led directly to the adoption and implementation of the new proxy metric on Netflix's experimentation platform, where it has been computed for 257 A/B tests as of this writing.

\subsection{Setup}
The data for this case study come from 123 A/B past tests at Netflix.  We chose these tests to be representative of the primary testing area that we wanted our new proxy metric to inform.  Specifically, we sampled tests of in-app interventions aimed at improving engagement, such as innovations to our recommendation and personalization algorithms.

Letting $Y$, $Z_0$, and $Z_1$ denote the north star metric, status quo proxy metric, and the new ``challenger'' proxy metric, respectively, we estimated the cumulative returns to a handful of candidate launch rules:\footnote{We also experimented with a number of other combinations (such as different tie-breakers and requiring statistical significance in both $Z_0$ and $Z_1$), which we omit for brevity.}

\begin{enumerate}
    \item Launch the arm with the highest statistically significant gain in $Z_0$.
    \item Launch the arm with the highest statistically significant gain in $Z_1$.
    \item Launch the arm with the highest gain in $Z_1$ given statistical significance in either $Z_0$ or $Z_1$.
\end{enumerate}

\subsection{Results}

\begin{figure}[ht]
    \centering
    \includegraphics[width=\linewidth]{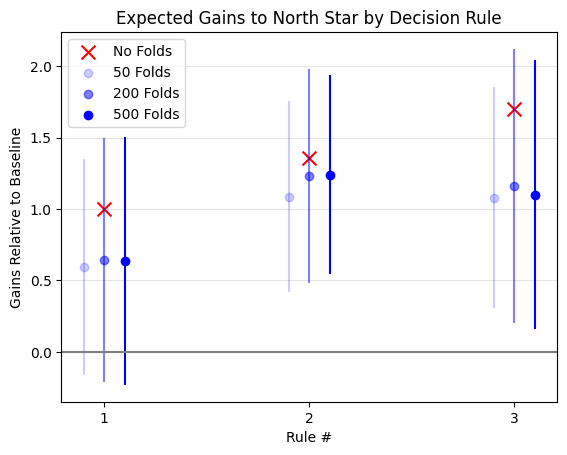}
    \caption{Estimated Cumulative Returns to Decision Rules Across 123 A/B Tests at Netflix.  We applied the methods in this paper to estimate the naive and cross-validated cumulative returns to three different decision rules at Netflix, contrasting two challenger rules to the status quo rule.  We find that a new proxy metric would have increased cumulative returns by over $33\%$ ($p$-values between $0.07$ and $0.09$).}
    \label{fig:empirics}
\end{figure}

Figure~\ref{fig:empirics} shows our results by plotting the naive (``No Folds'') estimate of the cumulative returns (red crosses) alongside our CV estimates under increasingly numerous folds (blue circles).  Error bars represent 95\% confidence intervals. To maintain confidentiality, we scale all values by baseline gains (the ``No Folds'' variant for Rule 1).

As expected, the naive estimator systematically overestimates the cumulative returns of past experiments due to the winner's curse. Although cross-validation could overshrink these estimates, the estimates remain stable as we increase the number of folds, indicating that the negative bias of the CV estimator is small in our application.

As the figure shows, our methods can also influence qualitative comparisons between decision rules. The naive ranking places Rule 3 far above Rule 2, while the CV estimator makes the two rules appear comparable. This is likely due to a combination of $Z_0$ having a more severe correlated errors problem than $Z_1$, as well as the slightly higher Type I error rate incurred by using an either-or launch rule. Altogether, the results demonstrate that the choice of estimator can meaningfully alter both quantitative (estimates of cumulative impact) and qualitative (choice of rule) business conclusions.

We find that both ``challenger'' rules would meaningfully outperform the status quo. In particular, we estimate that adopting Rule 3 would have increased north star metric gains in excess of 33\%, with $p$-values between 0.07 and 0.09 depending on the number of folds. Upon closer inspection, we found that these improvements stemmed largely from a few high-impact tests in which our new proxy metric disagreed sharply with the previous decision metric.  Individually, these tests were underpowered for the north star metric, but collectively they accounted for a substantial improvement.  Although marginally significant under traditional frequentist thresholds, the large gains at stake and modest implementation costs gave us confidence to recommend Rule 3. These estimated gains, as well as the explainability of our methodology, led directly to the adoption of this new launch rule by decision makers.  Our rule has since become recognized as the best practice for resolving metric disagreements in experiments.

\section{Conclusion}
\label{sec:conclusion}

This paper proposes to evaluate candidate decision rules using their cumulative returns to business north star metrics in past experiments. In our experience, cumulative returns are an intuitive evaluation metric that resonates with less technical stakeholders, making them a compelling quantity to report across a broad set of applications. Still, accurately estimating cumulative returns can be challenging when treatment effects are small relative to measurement error, as is often the case in digital experimentation, including at Netflix \citep{bibaut2024learning,larsen2024statistical, ejdemyr2024estimating}.

We demonstrate theoretically and empirically that cross-validation can effectively estimate the reward of a decision rule in such conditions.  In contrast, naive approaches that reuse the same data for both decision rule selection and evaluation lead to classical winner's curse biases, even when experimenters are optimizing for proxy metrics rather than north star metrics.

An implicit assumption underlying this work is the independence of treatment effects across tests: The difference in rewards between two treatment arms within an experiment cannot depend on the launch decisions of other experiments. This assumption could be violated if experiments implement redundant interventions or if previous launch decisions radically affect the design and impact of future tests. We believe that this is a reasonable approximation in our setting, which is characterized by many modest interventions across diverse domains and experiments that are typically planned far in advance.  That being said, one valuable direction for future research is to explore estimators that allow for dependence across experiments, for example by encoding dependencies between tests in cross-validation.  Relatedly, we also implicitly assume a degree of stationarity, such that rules optimized for past returns will also be useful for future experiments.  This assumption can be monitored by re-running the analysis periodically over time.

A second, more applied direction in which to extend this research is to evaluate other dimensions of experiment decision-making beyond the choice of outcome metric.  For example, recent papers have debated whether statistical significance thresholds should be more restrictive or permissive \citep{kohavi2024false, sudijono2024optimizing}.  The methodology outlined here provides a principled way for firms to evaluate the returns to different thresholds using actual experiments, and can easily be extended to include costs.

Lastly, another direction for future research is to \emph{learn} good decision rules (rather than pick the best rule from an \emph{ex ante}-defined set).  Here, an interesting trade-off is between optimality and interpretability:  a good decision rule should increase the returns to experimentation, but also be coherent enough to be trusted, understood, and remembered by decision-makers with varying statistical aptitudes.  We leave this to future work.

\section*{Acknowledgements}

For supporting this work, we thank Adrien Alexandre, Ai-Lei Sun, Danielle Rich, and Patric Glynn.  Code to replicate all figures except those in Section~\ref{sec:application}, which use proprietary data, can be found at \url{https://github.com/winston-chou/linear-proxy-metrics}.

\balance
\bibliography{bib}

\pagebreak

\section{Appendix}

\subsection{Derivation of Closed-Form Analysis}
\label{sec:app-closed}

To derive Equations \ref{eqn:target}-\ref{eqn:cv-target}, we use two properties of Gaussian random variables: bivariate conditional expectations,
\begin{equation}
    \mathbb{E}[A|B] = \mu_A + \frac{\sigma_{AB}}{\sigma^2_B}(B - \mu_B)
    \label{eqn:bivariate}
\end{equation}
and the Mills ratio:
\begin{equation}
    \mathbb{E}[A|A > 0] = \mu_A + \sigma_A \left( \frac{\phi(-\frac{\mu_A}{\sigma_A})}{1 - \Phi(-\frac{\mu_A}{\sigma_A})} \right).
    \label{eqn:mills-ratio}
\end{equation}

Applying the law of iterated expectations:
\begin{align*}
    \mathbb{E}[A | B > 0] &= \mathbb{E}[\mathbb{E}[A | B] | B > 0] = \mathbb{E}\left[ \mu_A + \frac{\sigma_{AB}}{\sigma_B^2} (B - \mu_B) | B > 0 \right] \\
    &= \mu_A + \frac{\sigma_{AB}}{\sigma^2_B} (\mathbb{E}[B | B > 0] - \mu_B) \\
    &= \mu_A + \frac{\sigma_{AB}}{\sigma_B} \left( \frac{\phi(-\frac{\mu_B}{\sigma_B})}{1 - \Phi(-\frac{\mu_B}{\sigma_B`})} \right),
\end{align*}
which is proportional to $\sigma_{AB} / \sigma_B$ when $\mu_A$ is equal to zero (as it is in our setup).

To derive Equation~\ref{eqn:target}, we set $A$ to the true treatment effect $\tau_Y^i$ and $B$ to the estimated treatment effect $\hat{\tau}_S^i$.  In Equation~\ref{eqn:naive-target}, we set $A$ to the estimated treatment effect $\hat{\tau}_Y^i$ and $B$ to the estimated treatment effect $\hat{\tau}_S^i$.  Lastly, in Equation~\ref{eqn:cv-target}, we set $A$ to the treatment effect estimated on the held-out fold and $B$ to the treatment effect on $S$ estimated on the remaining folds.

\subsection{Simulation Details}

Unless otherwise noted in the text, all simulation parameters are set to the following values:
\begin{eqnarray}
    N &=& 100 \\
    M &=& 1,000,000 \\
    \sigma_{\tau_Y} &=& 0.0001 \\
    \sigma_{\tau_S} &=& 0.01 \\
    \rho_\tau &=& 0.8 \\
    \sigma_Y &=& 0.10 \\
    \sigma_S &=& 10 \\
    \rho &=& 0.4 \\
    P &=& 10 \\
    \text{number of simulations} &=& 10,000.
\end{eqnarray}

\end{document}